\documentclass[11pt]{article}
\usepackage[utf8]{inputenc}
\usepackage{tcolorbox}
\usepackage{amsthm}
\usepackage{amsfonts}
\usepackage{amsmath}
\usepackage{amssymb}
\usepackage{colonequals}
\usepackage{epsfig}
\usepackage{url}
\newtheorem{theorem}{Theorem}
\newtheorem{corollary}[theorem]{Corollary}
\newtheorem{definition}[theorem]{Definition}
\newtheorem{lemma}[theorem]{Lemma}
\newtheorem{observation}[theorem]{Observation}
\newtheorem{problem}[theorem]{Problem}
\newtheorem{example}[theorem]{Example}
\newcommand{\GG}{{\cal G}}
\newcommand{\CC}{{\cal C}}

\renewcommand{\SS}{{\cal S}}

\newcommand{\gate}[3]{#1(#2\to#3)}

\newcommand{\VC}{\mathrm{VC}}

\title{Representation of short distances in structurally sparse graphs\footnote{Supported by the ERC-CZ project LL2005 (Algorithms and complexity within and beyond bounded expansion) of the Ministry of Education of Czech Republic.
Revised and extended with support of ERC Synergy grant DYNASNET no. 810115.}}
\author{Zden\v{e}k Dvo\v{r}\'ak\thanks{Computer Science Institute, Charles University, Prague, Czech Republic. E-mail: {\tt rakdver@iuuk.mff.cuni.cz}.}
}
\date{}

\begin{document}
\maketitle

\begin{abstract}
A partial orientation $\vec{H}$ of a graph $G$ is a \emph{weak $r$-guidance system} if for any two vertices
at distance at most $r$ in $G$, there exists a shortest path $P$ between them such that $\vec{H}$
directs all but one edge in $P$ towards this edge.  In case that $\vec{H}$ has bounded maximum outdegree,
this gives an efficient representation of shortest paths of length at most $r$ in $G$.
We show that graphs from many natural graph classes admit such weak guidance systems,
and study the algorithmic aspects of this notion.
\end{abstract}

\section{Introduction}

We consider the following general question: Given an undirected unweighted graph $G$, can short distances in $G$ be represented efficiently?
More precisely, the setting that interests us is as follows:
\begin{itemize}
\item $G$ is known to belong to some class $\GG$ of well-structured graphs (e.g., planar graphs, graphs of clique-width at most $6$, \ldots)
\item We are only interested in distances up to some fixed upper bound $r$.
\item We are allowed to preprocess $G$ in polynomial time; let $D$ denote the resulting data structure.
\item The data structure $D$ should enable us to efficiently answer the queries of the following form:
\begin{itemize}
\item Are two input vertices $u$ and $v$ at distance at most $r$ in $G$?
\end{itemize}
In case that the answer is positive, we may also want to determine the distance between $u$ and $v$,
and return a shortest path between them.
\end{itemize}
Note that we consider both $\GG$ and $r$ to be fixed parameters.  There are several criteria to consider:
\begin{itemize}
\item The time complexity of the preprocessing.
\item The time complexity of the queries.
\item The space complexity (the size of $D$).
\end{itemize}
Of course, there are some trade-offs between these criteria.  E.g., $D$ could store distances between all pairs of vertices, resulting in a relatively slow preprocessing time
and space complexity $\Theta(|V(G)|^2)$, but constant query time.  In this paper we consider a solution which still achieves constant query time
(depending only on $\GG$ and $r$), but is memory efficient in the sense that storing $D$ takes up about as much space as the graph $G$ itself.
To achieve this, $D$ will only consist of an orientation of $G$.

An \emph{orientation} of an undirected graph $G$ is a directed graph $\vec{H}$ such that for every $(u,v)\in E(\vec{H})$,
we have $uv\in E(G)$, and for every $uv\in E(G)$, at least one of $(u,v)$ and $(v,u)$ is a directed edge of $\vec{H}$.
Note that $\vec{H}$ can contain both $(u,v)$ and $(v,u)$, i.e., we allow an edge of $G$ to be directed
in both ways at the same time.  Let $B_{\vec{H}}(v,a)$ denote the set of vertices reachable in $\vec{H}$ from $v$ by a directed
path of length at most $a$.  An \emph{$r$-guidance system} is an orientation $\vec{H}$ such that
for any vertices $u,v\in V(G)$ at distance $\ell\le r$ in $G$, there exist non-negative integers $a$ and $b$ such that $a+b=\ell$
and $B_{\vec{H}}(u,a)\cap B_{\vec{H}}(v,b)\neq \emptyset$; i.e., there is a shortest path between $u$ and $v$ in $G$
whose edges are in $\vec{H}$ directed towards one of its vertices.  Note that if $\vec{H}$ has maximum outdegree at most $c$,
all such paths can be enumerated in time $O(c^r)$, and if $c$ is small, this enables us to find a shortest path between
a given pair of vertices (or verify that their distance is greater than $r$) efficiently.

The guidance systems were (without explicitly naming them) introduced by Kowalik and Kurowski~\cite{KowKur}, who
proved that they can be used to represent short distances in planar graphs, and more generally for every $F$,
in any graph avoiding $F$ as a topological minor.
As observed in~\cite{dvorlah}, essentially the same argument shows that graphs from even more general graph classes,
namely all classes with \emph{bounded expansion} and more generally all \emph{nowhere-dense} classes, admit guidance systems of bounded maximum outdegree. 
To state the result precisely, we need to introduce several definitions.

For a non-negative integer $s$, a graph $H$ is an \emph{$s$-shallow minor} of a graph $G$ if $H$ is obtained from $G$ by contracting pairwise-disjoint subgraphs,
each of radius at most $s$.  For a class $\GG$, let $\nabla_{\!s} \GG$ denote the class of all graphs $H$ that appear as $s$-shallow minors in graphs from $\GG$.
A class $\GG$ of graphs has \emph{bounded expansion} if for every $s\ge 0$ there exists $d_s$ such that every graph in $\nabla_{\!s} \GG$ has average degree at most $d_s$.
Even less restrictively, a class $\GG$ is \emph{nowhere-dense} if for every $s\ge 0$ there exists $d_s$ such that $K_{d_s}\not\in \nabla_{\!s} \GG$.
Examples of classes of graphs with bounded expansion include planar graphs and more generally all proper minor-closed classes, graphs with bounded maximum degree
and more generally all proper classes closed under topological minors, graphs drawn in the plane with $O(1)$ crossings on each edge, and many other classes
of sparse graphs; see~\cite{nesbook} for more details.

\begin{theorem}[Dvo\v{r}\'ak and Lahiri~\cite{dvorlah}]\label{thm-bexp}
Let $\GG$ be a class of graphs and $r$ a positive integer.
\begin{itemize}
\item If $\GG$ has bounded expansion, then there exists $c$ such that every graph $G\in \GG$ has
an $r$-guidance system of maximum outdegree at most $c$.  Moreover, such an $r$-guidance system can be found in time $O(|V(G)|)$.
\item If $\GG$ is nowhere-dense, then for every $\varepsilon>0$, there exists $c$ such that
every graph $G\in \GG$ has an $r$-guidance system of maximum outdegree at most $c|G|^\varepsilon$.
Moreover, such an $r$-guidance system can be found in time $O(|V(G)|^{1+\varepsilon})$.
\end{itemize}
\end{theorem}

A graph with an orientation of maximum outdegree at most $c$ necessarily has maximum average degree at most $2c$,
and thus it is $(2c+1)$-degenerate.  Hence, guidance systems of bounded maximum outdegree can only exist in sparse graphs.
This brings us to the main topic of our paper: \textbf{Does there exists a variant of the notion useful for dense graphs?}

Note that representing distance one by a guidance system forces us to orient all edges.  If we relax the notion to only represent distances 2, 3, \ldots, $r$, this may not be necessary.
A \emph{partial orientation} of a graph $G$ is a spanning directed subgraph of an orientation of $G$ (i.e., we allow some
edges not to be oriented in either direction).  An \emph{$r^+$-guidance system} is a partial orientation $\vec{H}$ of a graph $G$
such that for any vertices $u,v\in V(G)$ at distance $\ell$ in $G$, where $2\le \ell\le r$, there exist non-negative integers $a$ and $b$ such that $a+b=\ell$
and $B_{\vec{H}}(u,a)\cap B_{\vec{H}}(v,b)\neq \emptyset$.  Let us give a (trivial) example showing that there are
dense graphs admitting $r^+$-guidance systems.

\begin{example}\label{ex-univ}
Let $G$ be a graph containing a universal vertex $u$, and let $\vec{H}$ be the partial orientation obtained
by directing all edges incident with $u$ towards $u$.  Observe that for any positive integer $r$,
$\vec{H}$ is an $r^+$-guidance system in $G$ of maximum outdegree one.
\end{example}
However, there are some quite simple graphs that do not admit $r^+$-guidance systems of bounded outdegree.
For a graph $G$ and a positive integer $k$, let $G^k$ denote the \emph{$k$-distance power} of $G$, that is,
the graph with vertex set $V(G)$ and two vertices adjacent if and only if the distance between them in $G$ is at most $k$.

\begin{example}\label{example-power}
Let $T$ be the graph obtained from $K_{1,n}$ by subdividing every edge exactly twice, let $X$ be the
set of its leaves, and let $Y$ be the set of neighbors of the central vertex of degree $n$.  Let $G=T^2$.
Note that $Y$ induces a clique in $G$, and any two vertices of $X$ are joined by 
a unique path of length three using exactly one edge of this clique.  This implies that in any $3^+$-guidance system
for $G$, every edge of the clique on $Y$ must be directed in at least one direction, and thus some vertex of $Y$ has outdegree at least $(n-1)/2$.
\end{example}

This example highlights the fact that in dense graphs, we cannot afford to represent the shortest paths
by having all of their edges oriented, and motivates another generalization of guidance systems.
\begin{definition}
A \emph{weak $r$-guidance system}
is a partial orientation $\vec{H}$ of $G$ such that for any distinct vertices $u,v\in V(G)$ at distance $\ell\le r$ in $G$, there exist non-negative integers
$a$ and $b$ such that $a+b=\ell-1$ and $G$ contains an edge between $B_{\vec{H}}(u,a)$ and $B_{\vec{H}}(v,b)$; that is, there exists
a shortest path between $u$ and $v$ in $G$ such that all but one edge $e$ of this path is directed in $\vec{H}$
towards this exceptional edge $e$ (which may or may not be directed).
\end{definition}
In particular, an $r$-guidance system (or an $r^+$-guidance system) is also a weak $r$-guidance system.
Note that if the graph $G$ is represented so that we can in constant time test whether two vertices are adjacent, then a weak $r$-guidance system of maximum outdegree $c$ makes it possible
to find a shortest path between a given pair of vertices (or verify that their distance is greater than $r$) in time $O(c^{r-1})$.

The goal of this paper is to develop the theory of weak guidance systems; we show that several interesting
graph classes admit weak guidance systems of small maximum outdegree (constant, or logarithmic in the number of vertices),
address the algorithmic question of finding weak guidance systems efficiently, and describe an application of
the notion in approximation of distance variants of the independence and domination number.
On the negative side, we give examples of simple graph classes that do not admit weak guidance systems of small maximum outdegree.

The rest of the paper is organized as follows.
\begin{itemize}
\item In Section~\ref{sec-basic}, we give some basic properties of weak guidance systems,
including the fact that they behave well under the distance power operation.
\item In Section~\ref{sec-bexp}, we prove a result analogous to Theorem~\ref{thm-bexp},
showing that graphs from classes with structurally bounded expansion (i.e., definable
in classes with bounded expansion by first-order logic formulas) admit weak guidance
systems with bounded maximum outdegree.  We also give an analogous result for structurally nowhere-dense graph classes.
\item The results from Section~\ref{sec-basic} and~\ref{sec-bexp} do not provide polynomial-time algorithms to find
the weak guidance system if we are not provided with some additional information (e.g., in the case of
graph powers, if we are only given the graph $G^k$, but not the graph $G$ and its weak guidance system).
In Section~\ref{sec-algo}, we provide an approximation algorithm for this problem
that for an $n$-vertex graph which admits a weak guidance system of maximum outdegree $c$
returns one of maximum outdegree $O(c\log n)$.  We also provide an algorithm
that returns a weak guidance system of maximum outdegree $O(c\log c)$, assuming that
certain set systems have bounded VC-dimension, which is in particular the case for
classes with structurally bounded expansion studied in Section~\ref{sec-bexp}.
\item In Section~\ref{sec-appl}, we show an application of weak guidance systems
in design of approximation algorithms for distance independence and domination number.
\item In Section~\ref{sec-lb}, we consider several graph classes that do not admit weak
guidance systems of bounded maximum outdegree, specifically graphs of girth at least five
and large average degree, split graphs, and graphs of bounded clique-width.
\end{itemize}

\section{Basic properties of weak guidance systems}\label{sec-basic}

First, let us note that weak guidance systems enable us to circumvent the difficulty from Example~\ref{example-power}.
\begin{lemma}\label{lemma-power}
Let $G$ be a graph and let $k\ge 1$ and $c\ge 2$ be integers.
For any positive integer $r$, if $G$ has a weak $kr$-guidance system $\vec{H}$ of maximum outdegree at most $c$, then 
then $G^k$ has a weak $r$-guidance system $\vec{F}$ of maximum outdegree at most $2c^k$.
\end{lemma}
\begin{proof}
Let $\vec{F}$ be the partial orientation of $G^k$ containing exactly the directed edges $(u,v)$ such that $v\in B_{\vec{H}}(u,k)$.
Note that $\vec{F}$ has maximum outdegree at most
$$\sum_{\ell=1}^k c^\ell=c\cdot \frac{c^k-1}{c-1}<2c^k.$$

Suppose that the distance between vertices $x$ and $y$ in $G^k$ is $\ell\le r$.  Then the distance between $x$ and $y$ in $G$
is between $(\ell-1)k+1$ and $\ell k$, and since $\vec{H}$ is a weak $kr$-guidance system in $G$, there is a shortest
path $P$ between $x$ and $y$ in $G$ oriented in $\vec{H}$ towards an edge $e=x'y'$ of $P$, where $x,x',y',y$ appear in $P$ in order.  Let $x_0=x$, $x_1$, \ldots, $x_a$
be the maximal sequence of vertices of $P$ such that $x_i$ is at distance $ki$ from $x$ in $G$ and $e$ is contained in the subpath of $P$ between $x_a$ and $y$.
Analogously, let $y_0=y$, $y_1$, \ldots, $y_b$ be the maximal sequence of vertices of $P$ such that $y_i$ is at distance
$ki$ from $y$ in $G$ and $e$ is contained in the subpath of $P$ between $y_b$ and $x$.
Note that the distance between $x_a$ and $y_b$ in $G$ is at most $2k-1$, since each of them is at distance at most $k-1$ from $e$.
If the distance between $x_a$ and $y_b$ in $G$ is greater than $k$, then note that $\ell=a+b+2$ and
$y'\neq y_b$; let $z=y'$ and let $Q$ be the path $x_0\ldots x_azy_b\ldots y_0$.  Otherwise, $\ell=a+b+1$ and we let $z=y_b$ and
$Q=x_0\ldots x_ay_b\ldots y_0$.  This gives a shortest path $Q$ in $G^k$ directed in $\vec{F}$ towards its edge $x_az$.
\end{proof}

Let us remark that weak guidance systems are qualitatively different from guidance systems only in dense graphs,
as in degenerate graphs, a weak guidance system can be completed to a guidance system by directing the rest of the
edges while preserving the bounded maximum outdegree.
\begin{observation}\label{obs-degen}
If $G$ admits a weak $r$-guidance system of maximum
outdegree $c$ and $G$ is $t$-degenerate, then $G$ also admits an $r$-guidance system of maximum outdegree at most $c+t$.
\end{observation}

Finally, we give the following description of weak $r$-guidance systems, which we use often in the rest of the paper.
For vertices $u$ and $v$ of a graph $G$ at distance $\ell$, let $\gate{G}{u}{v}$ be the set of neighbors of $u$
at distance $\ell-1$ from $v$; i.e., $\gate{G}{u}{v}$ consists of all possible second vertices of shortest paths
from $u$ to $v$.
\begin{observation}\label{obs-char}
A partial orientation $\vec{H}$ of a graph $G$ is a weak $r$-guidance system if and only if the
following claim holds for all $u,v\in V(G)$ at distance $\ell$ in $G$, where $2\le \ell\le r$:
\begin{itemize}
\item[($\star$)] Either $u$ has an outneighbor in $\gate{G}{u}{v}$, or $v$ has an outneighbor in $\gate{G}{v}{u}$.
\end{itemize}
\end{observation}

\section{Weak guidance systems in structurally sparse graphs}\label{sec-bexp}

The standard way of generalizing the concepts of bounded expansion and nowhere-density to
dense graphs is through the notion of \emph{first-order transductions}, see e.g.~\cite{gajarsky2020new,gajarsky2020first,dreier2021lacon,nevsetvril2021rankwidth}.
For a positive integer $k$ and a graph $G$, let $kG$ denote the disjoint union of $k$ copies of $G$.
A \emph{transduction} $T$ consists of
\begin{itemize}
\item a positive integer $k$
\item a binary predicate symbol $M$ and unary predicate symbols $U_1$, \ldots, $U_s$, and
\item first-order formulas $\omega(x)$ and $\epsilon(x,y)$ with free variables $x$ (resp. $x$ and $y$)
using these predicate symbols and the binary predicate symbol $E$.
\end{itemize}
For graphs $H$ and $G$, we write $H\in T(G)$ if there exist sets $C_1,\ldots, C_s\subseteq V(kG)$
such that $V(H)$ consists exactly of the vertices $v\in V(kG)$ satisfying
$$kG,U_1\colonequals C_1,\ldots, U_s\colonequals C_s\models \omega(v)$$
and $E(H)$ consists exactly of the pairs $u,v\in V(H)$ such that
$$kG,U_1\colonequals C_1,\ldots, U_s\colonequals C_s\models \epsilon(u,v),$$
where the predicate symbol $E$ is interpreted as adjacency in $kG$ and $M$ is
interpreted as the equivalence between the $k$ copies of each vertex.

That is, a transduction allows us to blow up the graph by replicating each vertex a bounded
number of times, then non-deterministically color some vertices (via the predicates $U_1$, \ldots, $U_s$),
and finally define the vertices and edges of the new graph by a first-order formula.
As an example, if $T$ is the transduction with $k=1$, $s=0$, $\omega(x)=\text{true}$ and
$$\epsilon(x,y)=(x\neq y)\land (\exists z) (z=x\lor E(x,z))\land E(z,y),$$
then $H\in T(G)$ if and only if $H=G^2$.  Hence, the transduction operation generalizes
the graph power operations we considered in Lemma~\ref{lemma-power}.

For a class of graphs $\GG'$ and a transduction $T$, let $T(\GG')$ denote the class of all graphs $G$ such that
$G\in T(G')$ for some $G'\in \GG'$.
We say that a class of graphs $\GG$ has \emph{structurally bounded expansion} (resp., is \emph{structurally nowhere-dense})
if $\GG\subseteq T(\GG')$ for a transduction $T$ and a graph class $\GG'$ of bounded expansion (resp., being nowhere-dense).
The goal of this section is to show that such graph classes admit weak guidance systems with bounded maximum outdegree.

In preparation for that, let us start by considering the graph classes with bounded \emph{shrub-depth}.
The notion of shrub-depth was defined by Ganian et al.~\cite{shrub}
using the concept of \emph{tree models}.  For a positive integer $m$, an \emph{$m$-signature} is a function
$S:\mathbb{Z}^+\to 2^{[m]\times [m]}$ assigning a symmetric relation $S(i)$ to each $i>0$.  For a positive integer $d$,
an \emph{$(m,d)$-tree model} of a graph $G$ is a triple $(T,\varphi,S)$, where
\begin{itemize}
\item $T$ is a rooted tree with leaf set $V(G)$ and such that the length of every root-leaf path is $d$,
\item $\varphi:V(G)\to [m]$ assigns one of $m$ labels to each leaf,
\item $S$ is an $m$-signature, and
\item for every $u,v\in V(G)$, if $2i$ is the distance between $u$ and $v$ in $T$ (i.e., if $i$ is the distance from $u$ and $v$
to their nearest common ancestor in $T$), then $uv\in E(G)$ if and only if $(\varphi(u),\varphi(v))\in S(i)$.
\end{itemize}
A class $\GG$ of graphs has \emph{shrub-depth at most $d$} if for some positive integer $m$, every graph in $\GG$ has an $(m,d)$-model.

\begin{lemma}\label{lemma-orshrub}
For every class $\GG$ of graphs of bounded shrub-depth and every positive integer $r$, there exists a positive integer $c$ such that
every graph from $\GG$ has a weak $r$-guidance system of maximum outdegree at most $c$.
\end{lemma}
\begin{proof}
Let $m$ and $d$ be positive integers such that every graph $G\in \GG$ has an $(m,d)$-tree model $(T,\varphi,S)$.
Let $c=r^3m^r(d+1)^{r^2}d$.

For a positive integer $k$, a \emph{$k$-type} is a pair $(f,g)$ of functions $f:[k]\to [m]$ and $g:[k]^2\to \{0\}\cup [d]$.
The \emph{type} of a $k$-tuple $(v_1,\ldots, v_k)$ of vertices of $G$ is the $k$-type $(f,g)$ such that
$f(i)=\varphi(v_i)$ for $i\in [k]$ and $g(i,j)$ is half of the distance between $v_i$ and $v_j$ in $T$.
For each vertex $x\in V(T)$, each positive integer $k\le r$, and each $k$-type $t$,
if there exist a $k$-tuple $(v_1,\ldots,v_k)$ of leaves of $T$ with ancestor $x$ and of type $t$,
fix such a $k$-tuple $Q(x,t)=(v_1,\ldots, v_k)$ arbitrarily and let $A(x,t)=\{v_1,\ldots, v_k\}$; otherwise, let $A(x,t)=\emptyset$.
For each non-leaf vertex $y\in V(T)$, if $y$ has more than $r$ children $x$ such that $A(x,t)\neq\emptyset$,
then let $R(y,t)$ be a set of $r+1$ of them chosen arbitrarily; otherwise let $R(y,t)$ be the set of all
children $x$ of $y$ such that $A(x,t)\neq\emptyset$.  Let $B(y,t)=\bigcup_{x\in R(y,t)} A(x,t)$,
and let $B(y)$ be the union of $B(y,t)$ over all $k$-types $t$ with $k\le r$.

Let $\vec{H}$ be the partial orientation of $G$ containing exactly the edges $(u,v)$ such that $uv\in E(G)$
and $v\in B(y)$ for some ancestor $y$ of $u$ in $T$.  Clearly, $\vec{H}$ has maximum outdegree at most $c$.
Let us now argue that $\vec{H}$ is a weak $r$-guidance system.

Consider any vertices $u,v\in V(G)$ at distance $\ell$ in $G$, where $2\le \ell\le r$, and let $P=u_0u_1\ldots u_\ell$,
where $u_0=u$ and $u_\ell=v$, be a shortest path from $u$ to $v$ in $G$.  We will show that the condition ($\star$) from Observation~\ref{obs-char}
is satisfied for $u$ and $v$.  Let $y$ be the nearest common ancestor of $u$ and $u_1$ in $T$, let $X$ be the set of children of $y$
that have a descendant belonging to $V(P)$, and let $x_1$ be the child of $y$ whose descendant is $u_1$.
Suppose first that $v$ is not a descendant of $x_1$.  Let $Q$ be the tuple of vertices of $P$ that are descendants
of $x_1$ (in any order) and let $t$ be its type.  Since $|X|\le r+1$ and $A(x_1,t)\neq\emptyset$, there
exists $x'_1\in R(y,t)\setminus (X\setminus \{x_1\})$.  Let $Q'=Q(x'_1,t)$ and let $P'$ be obtained from $P$ by replacing
the vertices of $Q$ by the vertices of $Q'$.  Observe that since $Q$ and $Q'$ have the same type and the same common ancestors
with the other vertices of $P$, $P'$ is also a shortest path from $u$ to $v$ in $G$.  Moreover, the construction of $\vec{H}$
implies that the first edge of $P'$ is directed away from $u$, establishing the validity of the condition ($\star$) from Observation~\ref{obs-char}.

Hence, suppose that $v$ is a descendant of $x_1$.  In particular, this implies that $y$ is also the nearest common ancestor of $u$ and $v$.
Let $x_2$ be the child of $y$ whose descendant is $u$.  By symmetry, we can assume that $u_{\ell-1}$ is a descendant of $x_2$ as well.
Let $Q_1=(u_1,u_2,\ldots,u_k)$ be the maximal initial segment of $P-u$ consisting of descendants of $x_1$; we have $k<\ell-1$.
Let $t_1$ be the type of $Q_1$.  Since $|X|\le r+1$ and $A(x_1,t_1)\neq\emptyset$, there
exists $x''_1\in R(y,t_1)\setminus (X\setminus \{x_1\})$.  Let $Q'_1=Q(x'_1,t_1)$ and let $P'_1$ be obtained from $P$ by replacing
the vertices of $Q_1$ by the vertices of $Q'_1$.  Observe that since $Q$ and $Q'$ have the same type and the same common ancestors
with $u$ and $u_{k+1}$, $P'_1$ is also a shortest path from $u$ to $v$ in $G$.  Moreover, the construction of $\vec{H}$
implies that the first edge of $P'_1$ is directed away from $u$, establishing the validity of the condition ($\star$) from Observation~\ref{obs-char}.

We conclude that $\vec{H}$ is a weak $r$-guidance system.
\end{proof}

Crucially, the notions of structurally bounded expansion and structural nowhere-density can be characterized in terms
of \emph{bounded shrub-depth covers}.  A \emph{cover} of a graph $G$ is a system of subsets of $V(G)$.  Let $a$ be a positive integer.  A cover $\CC$ of $G$
is \emph{$a$-generic} if for every subset $A\subseteq V(G)$ of size at most $a$, there exists $C\in\CC$ such that $A\subseteq C$.
An \emph{$a$-generic bounded shrub-depth cover assignment} for a graph class $\GG$ is a function $\CC$ that to each graph $G\in \GG$
assigns an $a$-generic cover $\CC(G)$ such that the class $$\CC(\GG)=\{G[C]:G\in\GG,C\in\CC(G)\}$$ has bounded shrub-depth.

\begin{theorem}[Gajarsk{\'y} et al.~\cite{gajarsky2020first} and Dreier et al.~\cite{dreier2022treelike}]\label{thm-charcov}
Let $\GG$ be a class of graphs and let $a$ be a positive integer.
\begin{itemize}
\item If $\GG$ has structurally bounded expansion, then for some positive integer $k$, $\GG$ has
an $a$-generic bounded shrub-depth cover assignment $\CC$ such that $|\CC(G)|\le k$ for every $G\in \GG$.
\item If $\GG$ is structurally nowhere-dense and $\varepsilon>0$, then for some positive integer $k$, $\GG$ has
an $a$-generic bounded shrub-depth cover assignment $\CC$ such that $|\CC(G)|\le k|V(G)|^\varepsilon$ for every $G\in \GG$.
\end{itemize}
\end{theorem}

Together with Lemma~\ref{lemma-orshrub}, this gives the main result of this section.
\begin{corollary}\label{cor-main}
Let $\GG$ be a class of graphs and let $r$ be a positive integer.
\begin{itemize}
\item If $\GG$ has structurally bounded expansion, then for some positive integer $c$,
every graph in $\GG$ has a weak $r$-guidance system of maximum outdegree at most $c$.
\item If $\GG$ is structurally nowhere-dense and $\varepsilon>0$, then for some positive integer $c$,
every graph in $\GG$ has a weak $r$-guidance system of maximum outdegree at most $c|V(G)|^\varepsilon$.
\end{itemize}
\end{corollary}
\begin{proof}
Let $\CC$ be an $(r+1)$-generic bounded shrub-depth cover assignment and $k$ the corresponding constant from Theorem~\ref{thm-charcov}.
Let $c_0$ be the constant from Lemma~\ref{lemma-orshrub} for the class $\CC(\GG)$.  Let $c=kc_0$.

For any graph $G\in \GG$, let $\vec{H}$ be the union of the weak $r$-guidance systems of
the subgraphs $G[C]$ for $C\in\CC(G)$ obtained using Lemma~\ref{lemma-orshrub}.
Clearly, the maximum outdegree of $\vec{H}$ is at most $c$ if $\GG$ has structurally bounded expansion
and at most $c|V(G)|^\varepsilon$ if $\GG$ is structurally nowhere-dense.
Moreover, consider any vertices $u$ and $v$ at distance at most $r$ in $G$,
and let $P$ be a shortest path between them.
Since the cover $\CC(G)$ is $(r+1)$-generic, there exists $C\in \CC(G)$ such that $G[C]$ contains $P$.
Since $\vec{H}$ restricted to $C$ is a weak $r$-guidance system in $G[C]$,
there exists a shortest path between $u$ and $v$ in $G[C]$ (and thus also in $G$)
directed by $\vec{H}$ towards one of its edges.  We conclude that $\vec{H}$ is a weak $r$-guidance system in $G$.
\end{proof}

Let us remark that $r$-guidance systems can be used to characterize bounded expansion and nowhere-density.
\begin{lemma}
Let $\GG$ be a class of graphs closed under induced subgraphs.
\begin{itemize}
\item If there exists $c:\mathbb{Z}^+\to\mathbb{Z}^+$ such that for every positive integer $r$, every $G\in\GG$
has an $r$-guidance system of maximum outdegree at most $c(r)$, then $\GG$ has bounded
expansion.
\item If there exists $c:\mathbb{Z}^+\times \mathbb{R}^+\to\mathbb{Z}^+$ such that for every positive integer $r$ and for every $\varepsilon>0$, every $G\in\GG$
has an $r$-guidance system of maximum outdegree at most $c(r,\varepsilon)|V(G)|^\varepsilon$, then $\GG$ is nowhere-dense.
\end{itemize}
\end{lemma}
\begin{proof}
Suppose for a contradiction that $\GG$ is not nowhere-dense.  By assumptions, for every $\varepsilon>0$, every graph $G\in \GG$ has an orientation
with maximum outdegree at most $c(1,\varepsilon)|V(G)|^\varepsilon$, and thus the maximum average degree of subgraphs of $G$ is
at most $2c(1,\varepsilon)|V(G)|^\varepsilon$.  By~\cite[Theorem 6]{dvovrak2018induced}, there exists $r\ge 2$, a graph $G\in\GG$,
and a graph $H$ of average degree $d>2c(r,\varepsilon)|V(G)|^\varepsilon$ such that $G$ contains the graph $H'$
obtained from $H$ by subdividing each edge exactly $r-1$ times as an induced subgraph.  Since $\GG$ is closed under induced
subgraphs, we can assume $G=H'$.  Suppose $\vec{H}$ is an $r$-guidance system in $G$.
Then for every $uv\in E(H)$, the corresponding path $P_{uv}$ of length $r$ in $G$ contains an edge directed away from
$u$ or from $v$, and thus the average outdegree of the vertices of $H$ in $G$ is at least $|E(H)|/|V(H)|=d/2>c(r,\varepsilon)|V(G)|^\varepsilon$.
This contradicts the assumptions.

The argument for the bounded expansion case is analogous, using~\cite[Theorem 5]{dvovrak2018induced} instead of \cite[Theorem 6]{dvovrak2018induced}.
\end{proof}
Note that the assumption of being closed under induced subgraphs is needed, as seen by Example~\ref{ex-univ}: This example
together with Observation~\ref{obs-degen} shows that the class of graphs formed from cliques by subdividing each edge once
and adding a universal vertex afterwards admits an $r$-guidance system of maximum outdegree at most $4$
for every $r\ge 1$; but this class is not nowhere-dense.

It is tempting to ask whether weak $r$-guidance systems do not similarly characterize structurally bounded expansion or structural nowhere-density.  However, this is not the case.
We define a \emph{weak $\infty$-guidance system} to be a partial orientation that is a weak $r$-guidance system for every positive integer $r$.
\emph{Interval graphs} are the intersection graphs of sets of open intervals in the real line.
\begin{example}
Consider any interval graph $G$.  Let $\vec{H}$ be the partial orientation of $G$ obtained as follows.
For each $u\in V(G)$, let $v_1$ and $v_2$ be the neighbors of $u$ such that the right endpoint of the interval of $v_1$ is maximum
among all neighbors of $u$, and the left endpoint of the interval of $v_2$ is minimum among them.  Include in $\vec{H}$ the edges $(v,v_1)$ and $(v,v_2)$.
Then $\vec{H}$ is a weak $\infty$-guidance system in $G$ of maximum outdegree at most two.
\end{example}
The class of interval graphs is closed under induced subgraphs, but it is well-known not to be structurally nowhere-dense.

\section{Algorithmic aspects}\label{sec-algo}

Note that Theorem~\ref{thm-charcov} only gives a polynomial-time algorithm to obtain the covers if
we are given a graph $G'$ such that $G\in T(G')$, where $G'$ belongs to a bounded expansion/nowhere dense graph class.
If only $G$ is provided, it is currently not known how to obtain the covers efficiently.
Consequently, Corollary~\ref{cor-main} does not give an efficient algorithm to obtain weak guidance systems.
In this section, we address this issue, giving a polynomial-time algorithm that given an $n$-vertex
graph returns a weak guidance system whose maximum outdegree is worse than optimal only by an $O(\log n)$ factor,
and an improved approximation algorithm in case certain relevant set systems have bounded VC-dimension.

First, let us introduce one more relaxation of the guidance system notion.
A \emph{fractional orientation} of a graph $G$ is a function $p$ that assigns a non-negative
real number $p(u,v)$ to each pair $(u,v)$ of adjacent vertices of $G$.  The \emph{outdegree} $d^+_p(u)$
of a vertex $u$ in the fractional orientation $p$ is $\sum_{v:uv\in E(G)} p(u,v)$.
We say that $p$ is a \emph{fractional $r$-guidance system} if
for every $u,v\in V(G)$ at distance $\ell$, where $2\le \ell\le r$, we have
\begin{equation}\label{eq-proba}
\sum_{y\in \gate{G}{u}{v}} p(u,y)+\sum_{y\in \gate{G}{v}{u}} p(v,y)\ge 1.
\end{equation}

By Observation~\ref{obs-char}, weak guidance systems can naturally be interpreted as fractional guidance systems.
\begin{observation}\label{obs-tofra}
Suppose $\vec{H}$ is a weak $r$-guidance system in a graph $G$, of maximum outdegree $c$.
Let us define $p(u,v)=1$ for every $(u,v)\in E(\vec{H})$ and $p(u,v)=0$ for every $uv\in E(G)$
such that $(u,v)\not\in E(\vec{H})$.  Then $p$ is a fractional $r$-guidance system of maximum outdegree $c$.
\end{observation}

Moreover, an optimal fractional guidance system can be constructed through linear programming.

\begin{lemma}\label{lem-sollp}
If a graph $G$ has a weak $r$-guidance system of maximum outdegree $c_0$,
we can find a fractional $r$-guidance system of maximum outdegree at most $c_0$ in $G$ in polynomial time.
\end{lemma}
\begin{proof}
Let $p$ be an optimal solution to the following linear program.
\begin{align*}
p(u,v)&\ge 0&\text{ for every $(u,v)$ s.t. $uv\in E(G)$}\\
\sum_{v:uv\in E(G)} p(u,v)&\le c&\text{ for every $u\in V(G)$}\\
\sum_{y\in \gate{G}{u}{v}} p(u,y)+\sum_{y\in \gate{G}{v}{u}} p(v,y)&\ge 1&\text{ for every $u,v\in V(G)$ at distance between $2$ and $r$}\\
\text{minimize }&c
\end{align*}
Then $p$ is a fractional $r$-guidance system of maximum outdegree at most $c$ in $G$,
and $c\le c_0$ by Observation~\ref{obs-tofra}.

Note that the sets $\gate{G}{u}{v}$ can be computed in polynomial time by first computing the distances between
all pairs of vertices of $G$.  The linear program has $O(n^2)$ variables and constraints, and thus its
optimal solution can be found in polynomial time.  
\end{proof}

Fractional $r$-guidance systems can be directly used to test presence of shortest paths,
with a small probability of error.
Let $p$ be a fractional $r$-guidance system in a graph $G$.
If $u$ is a non-isolated vertex of $G$, then by a \emph{$p$-random neighbor of $u$}, we mean a neighbor of $u$ selected at random, with the
probability that a neighbor $v$ is selected being $p(u,v)/d^+_p(u)$; if $d^+_p(u)=0$, the probability is $1/\deg u$, instead.
For distinct vertices $u$ and $v$ and a positive integer $r$, a \emph{random $(p,r)$-exploration} between $u$ and $v$
is a random pair of walks $(P_u,P_v)$ from $u$ and $v$ selected as follows:
\begin{itemize}
\item If $uv\in E(G)$, then $P_u=uv$ and $P_v=v$.
\item Otherwise, if $r=1$ or $u$ or $v$ is an isolated vertex, then $P_u=u$ and $P_v=v$.
\item Otherwise, let $x\in \{u,v\}$ be selected uniformly at random, and let $y$ be a $p$-random neighbor of $x$;
\begin{itemize}
\item if $x=u$, then select a random $(p,r-1)$-exploration $(P_y,P_v)$ between $y$ and $v$ and let $P_u$ be the concatenation of $uy$ and $P_y$, and
\item if $x=v$, then select a random $(p,r-1)$-exploration $(P_u,P_y)$ between $u$ and $y$ and let $P_v$ be the concatenation of $vy$ and $P_y$.
\end{itemize}
\end{itemize}

\begin{observation}\label{obs-prob}
Suppose $p$ is a fractional $r$-guidance system in a graph $G$, of maximum outdegree $c$.
Let $u$ and $v$ be distinct vertices of $G$ at distance at most $r$, and let $(P_u,P_v)$
be a random $(p,r)$-exploration between $u$ and $v$. The probability that $P_u\cup P_v$
is a shortest path between $u$ and $v$ in $G$ is at least $(4c)^{-(r-1)}$.
\end{observation}
\begin{proof}
We prove the claim by induction on the distance $\ell$ between $u$ and $v$ that the probability is at least
$(4c)^{-(\ell-1)}$.  If $\ell=1$, then $P_u=uv$ and $P_v=v$ with probability $1$.  Hence, suppose that $\ell\ge 2$.
By (\ref{eq-proba}) and symmetry, we can assume that
$$\sum_{y\in \gate{G}{u}{v}} p(u,y)\ge \frac{1}{2},$$
and thus
$$\frac{\sum_{y\in \gate{G}{u}{v}} p(u,y)}{d^+_p(u)}\ge \frac{1}{2c}.$$
Hence, with probability at least $\frac{1}{4c}$, when choosing $(P_u,P_v)$, we pick $x=u$ and $y\in \gate{G}{u}{v}$,
so that the distance between $y$ and $v$ is $\ell-1$.  By the induction hypothesis, the probability that $P_y\cup P_v$
is a shortest path between $y$ and $v$, and thus $P_u\cup P_v$ is a shortest path between $u$ and $v$, is then
at least $(4c)^{-(\ell-2)}$.  The result follows by multiplying these probabilities.
\end{proof}

Note that for Observation~\ref{obs-prob} to be practically useful, we would need
a representation of $p$ that enables us to choose a $p$-random neighbor efficiently;
in that case, we could iterate $k(4c)^{r-1}$ times the procedure from Observation~\ref{obs-prob}
to find the shortest path between $u$ and $v$ (or decide that the distance between them is greater than $r$)
with error probability at most $e^{-k}$.
More interestingly, we can turn a fractional $r$-guidance system to a weak $r$-guidance system with a logarithmic loss in the maximum degree.

\begin{lemma}\label{lemma-defrac}
Let $c$ be a positive real number, let $n$ be a positive integer, and let $m=\lceil 4c\log n\rceil$.
Suppose $p$ is a fractional $r$-guidance system in a graph $G$, with maximum outdegree at most $c$.
There exists an algorithm that in polynomial time
returns a weak $r$-guidance system $\vec{H}$ in $G$ with maximum outdegree at most $m$.
\end{lemma}
\begin{proof}
Let us say that pair $\{u,v\}$ of vertices is \emph{dissatisfied} by a partial orientation $\vec{F}$
if the distance $\ell$ between $u$ and $v$ satisfies $2\le\ell\le r$ and $\vec{F}$ contains neither 
an edge from $u$ to $\gate{G}{u}{v}$ nor an edge from $v$ to $\gate{G}{v}{u}$.  By Observation~\ref{obs-char},
$\vec{F}$ is a weak $r$-guidance system if and only if there are no dissatisfied pairs.

Let $X$ be any set of pairs of vertices of $G$ at distance between $2$ and~$r$.
Let $\vec{F}$ be a random partial orientation of $G$ obtained by, for each non-isolated vertex $z$ of $G$,
choosing a random $p$-neighbor $z'$ and adding the edge $(z,z')$.   Clearly, $\vec{F}$ has maximum outdegree at most one.
Moreover, consider any $\{u,v\}\in X$.  By (\ref{eq-proba}) and symmetry, we can assume that
$$\sum_{y\in \gate{G}{u}{v}} p(u,y)\ge 1/2.$$
Hence, the probability that $u'\in \gate{G}{u}{v}$ (and thus $\{u,v\}$ is not dissatisfied in $\vec{F}$) is at least $\tfrac{1}{2c}$.
By the linearity of expectation, the expected number of dissatisfied pairs in $X$ is at most $\bigl(1-\tfrac{1}{2c}\bigr)|X|$.

Moreover, we can use the method of conditional probabilities to derandomize this procedure and to deterministically
construct a partial orientation $\vec{F}$ of $G$ of maximum outdegree at most one such that the number of pairs in $X$
dissatisfied by $\vec{F}$ is at most $\bigl(1-\tfrac{1}{2c}\bigr)|X|$. Indeed, we can select the outneighbors
one by one, always maintaining the invariant (initially satisfied by the computation from the previous paragraph)
that the expected number of pairs in $X$ dissatisfied by the orientation obtained by choosing the remaining outneighbors
as random $p$-neighbors is at most $\bigl(1-\tfrac{1}{2c}\bigr)|X|$.  To do so, when processing a vertex $u$,
we only need to be able to compute this expected number after each possible choice of the outneighbor of $u$,
which is straightforward due to the linearity of expectation.

Now, to obtain $\vec{H}$, we let $X_0$ be the set of all pairs of vertices whose distance is between $2$ and $r$ in $G$.
Then, for $i=1,\ldots, m$, we use the procedure described in the previous paragraph to find a partial orientation $\vec{F}_i$
of maximum outdegree at most one so that the set $X_i$ of pairs from $X_{i-1}$ dissatisfied by $\vec{F}_i$ has size at most
$\bigl(1-\tfrac{1}{2c}\bigr)|X_{i-1}|$.  Note that 
$$|X_m|\le \bigl(1-\tfrac{1}{2c}\bigr)^m|X_0|\le \frac{|X_0|}{n^2}<1,$$
and thus $X_m=\emptyset$.  Consequently, no pair is dissatisfied by
$$\vec{H}=\bigcup_{i=1}^m\vec{F}_i,$$
and thus $\vec{H}$ is the desired weak $r$-guidance system in $G$.
\end{proof}

Combining Lemmas~\ref{lem-sollp} and~\ref{lemma-defrac}, we obtain the following claim.
\begin{corollary}\label{cor-lp}
There exists an algorithm that, for an input $n$-vertex graph $G$ that admits a weak $r$-guidance system of maximum outdegree at most $c$,
outputs in polynomial time a weak $r$-guidance system of maximum outdegree $O(c\log n)$.
\end{corollary}

Let us remark that the logarithmic loss in Corollary~\ref{cor-lp} cannot be avoided in general.
For positive integers $a$ and $k$, let $m=k2^{k+1}$ and let $G_{a,k}$ be the random graph obtained as follows.
We start with a random bipartite graph with parts $L$ of size $a$ and $R$ of size $ma$, with each vertex
of $L$ being adjacent to each vertex of $R$ independently with probability $1/2$.  We then divide $R$
into $m$ parts $R_1$, \ldots, $R_m$ of size $a$ arbitrarily, and for $i=1,\ldots,m$, we add
a vertex $x_i$ adjacent to all vertices of $R_i$.

\begin{lemma}\label{lemma-gak}
There exists an integer $a_0$ such that for every $a\ge a_0$ and $k\le \log a$, with positive probability
\begin{itemize}
\item $G_{a,k}$ has a fractional $2$-guidance system with
maximum outdegree at most $3$, and
\item $G_{a,k}$ does not have a weak $2$-guidance system with maximum outdegree at most $k$.
\end{itemize}
\end{lemma}
\begin{proof}
Let us use the notation from the definition of the graph $G_{a,k}$.  Note that
\begin{itemize}
\item for $i=1,\ldots,m$ and $v\in L$, the expected number of neighbors of $v$ in $R_i$
is $a/2$, and by Chernoff inequality, the probability that $v$ has less than $a/3$ neighbors
in $R_i$ is less than $\exp(-a/36)$.
\item for distinct vertices $u,v\in R$, the expected number of common neighbors of $u$ and $v$ in $L$
is $a/4$, and by Chernoff inequality, the probability that $u$ and $v$ have less than $a/5$
common neighbors in $L$ is less than $\exp(-a/200)$,
\item for distinct $u,v\in L$, the probability that $u$ and $v$ have less than $a/5$ common neighbors
in $R_1$ is also less than $\exp(-a/200)$, and
\item for $i\in 1,\ldots,m$ and a $k$-tuple $K$ of vertices of $R_i$, the expected number
of vertices of $L$ with no neighbor in $K$ is $2^{-k}a$, and by Chernoff inequality,
the probability that the number of such vertices is at most $2^{-k-1}a$ is
at most $\exp\bigl(-2^{-k-3}a\bigr)$.
\end{itemize}
Hence, the probability that any of these events occurs is less than
$$ma\cdot \exp(-a/36)+(m^2+1)a^2\cdot \exp(-a/200)+ma^k\cdot \exp\bigl(-2^{-k-3}a\bigr)<1$$
if $a$ is sufficiently large (and using the assumption that $k\le \log a$; note that the basis
of the logarithm is $e$, and thus $2^k\le a^{\log 2}\ll a$).
Hence, with positive probability,
\begin{itemize}
\item for $i=1,\ldots,m$, each vertex $v\in L$ has at least $a/3$ neighbors in~$R_i$,
\item any distinct vertices $u,v\in R$ have at least $a/5$ common neighbors in~$L$,
\item any distinct vertices $u,v\in L$ have at least $a/5$ common neighbors in~$R_1$, and
\item for $i\in 1,\ldots,m$ and for every $k$-tuple $K$ of vertices of $R_i$, more than
$2^{-k-1}a$ vertices of $L$ have no neighbor in $K$.
\end{itemize}

Let us define a fractional orientation $p$
of $G_{a,k}$ as follows:
\begin{itemize}
\item For $i=1,\ldots,m$ and $v\in R_i$, we set $p(x_i,v)=3/a$,
\item for each adjacent $u\in R$ and $z\in L$, we set $p(u,z)=2.5/a$, and
\item for each adjacent $z\in L$ and $u\in R_1$, we set $p(z,u)=2.5/a$;
\end{itemize}
$p$ is $0$ everywhere else.  Note that this fractional orientation
has maximum outdegree at most 3, since $\deg x_i=|R_i|=a$,
the number of neighbors of $u\in R$ in $L$ is at most $|L|=a$,
and the number of neighbors of $z\in L$ in $R_1$ is at most $|R_1|=a$.
Consider now any vertices $x,y\in V(G_{a,k})$ at distance exactly two from one another.  Note that $G_{a,k}$ is bipartite,
and thus either $x,y\in R$, or $x,y\in V(G_{a,k})\setminus R$.  There are the following cases:
\begin{itemize}
\item One of $x$ and $y$ belongs to $\{x_1,\ldots,x_m\}$, say $x=x_i$.  Then $y$ necessarily belongs to $L$,
and $y$ has at least $a/3$ neighbors in $R_i$.  Hence, $|\gate{G_{a,k}}{x}{y}|\ge a/3$
and
$$\sum_{z\in \gate{G_{a,k}}{x}{y}} p(x,z)\ge a/3 \cdot 3/a=1.$$
\item Both $x$ and $y$ belong to $L$.  Since $x$ and $y$ have at least $a/5$ common neighbors in $R_1$,
we have $|\gate{G_{a,k}}{x}{y}\cap R_1|=|\gate{G_{a,k}}{y}{x}\cap R_1|\ge a/5$, and
$$\sum_{z\in \gate{G_{a,k}}{x}{y}} p(x,z)+\sum_{z\in \gate{G_{a,k}}{y}{x}} p(y,z)\ge 2\cdot a/5 \cdot 2.5/a=1.$$
\item Similarly, if $x,y\in R$, then $x$ and $y$ have at least $a/5$ common neighbors in $L$,
and $$\sum_{z\in \gate{G_{a,k}}{x}{y}} p(x,z)+\sum_{z\in \gate{G_{a,k}}{y}{x}} p(y,z)\ge 2\cdot a/5 \cdot 2.5/a=1.$$
\end{itemize}
Therefore, $p$ is a fractional $2$-guidance system for $G_{a,k}$.

Consider now any partial orientation $\vec{H}$ of $G_{a,k}$ with maximum outdegree at most $k$.
Then each vertex $v\in L$ has an outneighbor in $R_i$ for at most $k$ choices of $i$,
and thus there exists $i\in \{1,\ldots,m\}$ such that at least $(1-k/m)a=\bigl(1-2^{-k-1}\bigr)a$ vertices
of $L$ have no outneighbor in $R_i$.  Let $K$ be a $k$-tuple of vertices of $R_i$ containing all outneighbors
of $x_i$.  More than $2^{-k-1}a$ vertices of $L$ have no neighbor in $K$, and thus there exists a vertex $v\in L$
with no outneighbor in $R_i$ and no neighbor in $K$.  However, $x_i$ and $v$ are at distance $2$,
yet neither $x_i$ nor $v$ has an outneighbor in $\gate{G_{a,k}}{x_i}{v}=\gate{G_{a,k}}{v}{x_i}\subseteq R_i\setminus K$.
Hence $\vec{H}$ is not a weak $2$-guidance system.  Consequently, every weak $2$-guidance system for $G_{a,k}$
must have maximum outdegree greater than $k$.
\end{proof}

Note that if we set $k=\lfloor \log a\rfloor$, we have
$$n=|V(G_{a,k})|\le (m+1)(a+1)\le \bigl(k2^{k+1}+1\bigr)\cdot(\exp(k+1)+1)\le \exp(O(k)),$$
and thus Lemma~\ref{lemma-gak} gives examples of graphs with an arbitrarily large number $n$ of vertices and
a fractional $2$-guidance system of maximum outdegree at most $3$ such that every weak $2$-guidance system
has maximum outdegree $\Omega(\log n)$.

However, we can do better in case the VC-dimension of relevant systems is bounded.  Recall that a system $\SS$
of subsets of a set $X$ \emph{shatters} a set $A\subseteq X$ if $\{A\cap S:S\in\SS\}$ contains all subsets of $A$,
and that the \emph{VC-dimension} of $\SS$ is the size of the largest subset of $X$ shattered by $\SS$.
The key property of systems with bounded VC-dimension is that they admit efficient (randomized) approximation
for smallest hitting set in terms of the size of the smallest fractional hitting set (a \emph{hitting set} for $\SS$
is a subset of $X$ intersecting all elements of $\SS$, and a \emph{fractional hitting set} is a function
$w:X\to \mathbb{R}_0^+$ such that, defining $w(A)=\sum_{x\in A} w(x)$ for each subset $A$ of $X$, 
each element $S\in\SS$ satisfies $w(S)\ge 1$; the size of the fractional hitting set $w$ is $w(X)$).
For the following standard result, see e.g.~\cite{pach2011combinatorial}.

\begin{theorem}\label{thm-approx}
There exists a polynomial-time randomized algorithm that, given a system $\SS$ of subsets of a set $X$ of VC-dimension at most $d$
and a fractional hitting set $w$ of size $s$, with probability at least $1/2$ returns a hitting set for $\SS$
of size $O(ds\log s)$.
\end{theorem}

For a graph $G$, integer $r\ge 2$, and vertex $u\in V(G)$, let $\VC(G,r,u)$ denote the VC-dimension of the
system
$$\{\gate{G}{u}{v}:v\in V(G), 2\le d_G(u,v)\le r\},$$
and let $\VC(G,r)=\max_{u\in V(G)} \VC(G,r,u)$.
\begin{theorem}\label{thm-corlp-vc}
There exists a polynomial-time randomized algorithm that, for an input $n$-vertex graph $G$ that admits a weak
$r$-guidance system of maximum outdegree at most $c$, with probability at least $1/2$ outputs a weak $r$-guidance
system of maximum outdegree $O(\VC(G,r)\cdot c\log c)$.
\end{theorem}
\begin{proof}
Let $p$ be a fractional $r$-guidance system of maximum outdegree at most $c$ in $G$ found using Lemma~\ref{lem-sollp}.
For each $u\in V(G)$, let $R_u$ be the set of vertices $v\in V(G)$ such that $2\le d_G(u,v)\le r$
and
$$\sum_{z\in \gate{G}{u}{v}} p(u,z)\ge 1/2.$$
Since $p$ is a fractional $r$-guidance system, for each $u,v\in V(G)$ such that $2\le d_G(u,v)\le r$,
we have $v\in R_u$ or $u\in R_v$.

Let $\SS_u$ be the system $\{\gate{G}{u}{v}:v\in R_u\}$ of subsets of the set $N_G(u)$ of neighbors of $u$.
For $z\in N_G(u)$, let us define $w(z)=2p(u,z)$.  By the choice of $R_u$, we have $w(S)\ge 1$ for each $S\in\SS_u$,
and thus $w$ is a fractional hitting set for $\SS_u$.  Moreover, $w(N_G(u))\le 2c$, since the maximum outdegree of $p$ is at most $c$.
The VC-dimension of $\SS_u$ is at most $\VC(G,r,u)\le \VC(G,r)$, and thus we can by Theorem~\ref{thm-approx} find
a hitting set $H_u\subseteq N_G(u)$ for $\SS_u$ of size $O(\VC(G,r)\cdot c\log c)$; note that we iterate the algorithm $\Omega(|V(G)|)$
times to make the probability of error less than $\tfrac{1}{2|V(G)|}$, and thus we find a valid hitting set for each $u\in V(G)$
with probability at least $1/2$.

Let us now define a partial orientation $\vec{G}$ of $G$ by, for each $u\in V(G)$, directing the edges from $u$ to $H_u$.
Clearly, $\vec{G}$ has maximum outdegree $O(\VC(G,r)\cdot c\log c)$.  Moreover, consider any $u,v\in V(G)$ such that $2\le d_G(u,v)\le r$.
By symmetry, we can assume that $v\in R_u$, and thus $H_u$ intersects the set $\gate{G}{u}{v}\in\SS_u$.
Hence, $u$ has an outneighbor in $\gate{G}{u}{v}$.  By Observation~\ref{obs-char}, we conclude that $\vec{G}$ is
a weak $r$-guidance system for $G$.
\end{proof}

In particular, this is useful for structurally nowhere-dense classes (and in particular
for classes with structurally bounded expansion), as follows from the fact
that first-order definable sets in graphs from these classes have bounded VC-dimension.
More precisely, for a first-order formula $\psi(\vec{x},\vec{y})$ with two groups $\vec{x}$ and $\vec{y}$
of free variables, a graph $G$, and a $|\vec{x}|$-tuple $\vec{u}$ of vertices of $G$,
let $S_{\psi,G}(\vec{u})$ be the set of $|\vec{y}|$-tuples $\vec{v}$ of vertices of $G$
such that $G\models\psi(\vec{u},\vec{v})$, and let $\SS_{\psi,G}$ be the system
$$\{S_{\psi,G}(\vec{u}):\vec{u}\in V(G)^{|\vec{x}|}\}$$
of sets of $|\vec{y}|$-tuples of vertices of $G$.  The following bound follows from
the results of Adler and Adler~\cite{adl}, see also~\cite{numtypes} for a more precise
bounds and the discussion of the possibility to introduce vertex and edge colors (unary and binary predicates
from the statement of the theorem).
\begin{theorem}\label{thm-vc}
For every nowhere-dense graph class $\GG$ and a first-order formula $\psi(\vec{x},\vec{y})$
using unary predicate symbols $U_1$, \ldots, $U_s$ and binary predicate symbols $E_1$, \ldots, $E_t$,
there exists a constant $d$ such that the following claim holds.
Consider any graph $G\in \GG$, and interpret $U_i$ for $i\in \{1,\ldots, s\}$ as a subset of $V(G)$
and $E_j$ for $j\in\{1,\ldots,t\}$ as a subset of $E(G)$.
Then the system $\SS_{\psi,G}$ has VC-dimension at most $d$.
\end{theorem}

This easily gives the following consequence.

\begin{lemma}\label{lemma-vc}
For every structurally nowhere-dense class $\GG$ of graphs and every integer $r\ge 2$, there exists a constant $d$
such that $\VC(G,r)\le d$ for every graph $G\in\GG$.
\end{lemma}
\begin{proof}
Since $\GG$ is structurally nowhere-dense, there exists a nowhere-dense class $\GG_0$
and a transduction $T=(k,M,U_1,\ldots,U_s,\omega,\epsilon)$ such that for each $G\in \GG$
there exists a graph $H\in \GG_0$ such that $G\in T(H)$; let $C_1^G$, \ldots, $C_s^G$
be the corresponding subsets of $V(H)$ used to interpret $U_1$, \ldots, $U_s$.

For $H\in\GG_0$, let $(kH)'$ be the graph obtained from the disjoint union of $k$ copies of $G$ by adding a clique
on each $k$-tuple of vertices corresponding to the same vertex of $H$, and let $M_H$ be the set of the
edges of these cliques.  Also, let $E_H$ be the set of edges of $kH$.
Let $\GG_1=\{(kH)':H\in\GG_0\}$.  Since $(kH)'$ is a subgraph of the lexicographic
product of $H$ with a clique of bounded size and $\GG_0$ is nowhere-dense, the class $\GG_1$
is nowhere-dense as well~\cite{nesbook}.

Note that there exists a first-order formula $\psi_r(x_1,x_2,y)$ with three free variables such that
for each $u,v\in V(G)$ satisfying $2\le d_G(u,v)\le r$ and $z\in V(G)$, $G\models \psi_r(u,v,z)$
if and only if $z\in\gate{G}{u}{v}$.  Let $\psi'_r$ be the formula obtained from $\psi_r$
by restricting the quantification to vertices satisfying $\omega$ and replacing each usage
of the adjacency predicate by $\epsilon$.  Clearly, if $G\in T(H)$, then
$$G\models \psi_r(u,v,z)\text{ iff }(kH)',U_1\colonequals C_1^G,\ldots, U_s\colonequals C^G_s,E\colonequals E_H,M\colonequals M_H\models\psi_r(u,v,z).$$
Therefore, with the interpretation of the unary and binary symbols as above,
$\SS_{\psi_r,G}$ is a subset of $\{S\cap V(G):S\in \SS_{\psi'_r,(kH)'}\}$, and thus the VC-dimension of
$\SS_{\psi_r,G}$ is at most as large as the VC-dimension of $\SS_{\psi'_r,(kH)'}$.  Since $(kH)'\in\GG_1$
and $\GG_1$ is nowhere-dense, Theorem~\ref{thm-vc} implies that this VC-dimension is bounded.
\end{proof}

Hence, Theorem~\ref{thm-corlp-vc} gives the following algorithmic form of Corollary~\ref{cor-main}.

\begin{corollary}\label{cor-maina}
Let $\GG$ be a class of graphs and let $r$ be a positive integer.
\begin{itemize}
\item If $\GG$ has structurally bounded expansion, then there exists $c$ 
and a randomized algorithm that for an input $n$-vertex graph $G\in\GG$
outputs in polynomial time with probability at least $1/2$ a weak $r$-guidance system of maximum outdegree at most $c$.
\item If $\GG$ is structurally nowhere-dense and $\varepsilon>0$, then there exists $c$ 
and a randomized algorithm that for an input $n$-vertex graph $G\in\GG$
outputs in polynomial time with probability at least $1/2$ a weak $r$-guidance system of maximum outdegree at most $cn^\varepsilon$.
\end{itemize}
\end{corollary}

\section{Distance domination and independence number}\label{sec-appl}

For a positive integer $r$, a set $S$ of vertices of a graph $G$ is \emph{$r$-dominating}
if every vertex of $G$ is at distance at most $r$ from $S$, and \emph{$r$-independent}
if distinct vertices of $S$ are at distance greater than $r$ from one another.
Let $\gamma_r(G)$ denote the smallest size of an $r$-dominating set in $G$,
and $\alpha_r(G)$ the largest size of an $r$-independent set in $G$.  Observe that
if $D$ is an $r$-dominating and $A$ a $2r$-independent set in $G$, then
every vertex of $D$ is at distance at most $r$ from at most one vertex of $A$,
and since every vertex of $A$ is at distance at most $r$ from $D$, we have
$|A|\le |D|$.  Consequently, $\alpha_{2r}(G)\le \gamma_r(G)$.
In general, the converse inequality does not hold and it is not even possible to bound $\gamma_r(G)$ by a function of $\alpha_{2r}(G)$;
however, as shown in~\cite{apxdomin}, if $G$ is from a class of graphs with bounded expansion,
then $\gamma_r(G)=O(\alpha_{2r}(G))$.  A small variation of the argument gives the following
stronger claim.
\begin{lemma}\label{lemma-guidac}
For all positive integers $c$ and $r$, there exists a linear-time algorithm
that, given a graph $G$ together with its $2r$-guidance system of maximum outdegree less than $c$,
returns an $r$-dominating set $D$ and a $2r$-independent set $A$ in $G$ such that
$|D|\le c^2|A|$.
\end{lemma}
Note that this implies that $\gamma_r(G)\le |D|\le c^2\gamma_r(G)$ and $\tfrac{1}{c^2}\alpha_{2r}(G)\le |A|\le\alpha_{2r}(G)$,
and thus this gives a linear-time algorithm to approximate both the $r$-domination and the $2r$-independence number of $G$
within the constant factor $c^2$.  The presence of a weak $2r$-guidance system of bounded outdegree is not by itself sufficient
to ensure a similar result.

\begin{example}
Let $\vec{K}$ be a random orientation of the clique with vertex set $\{1, \ldots,n\}$
(for each edge, choose direction uniformly independently at random).  Let $G$ be the graph obtained from $\vec{K}$ as
follows: We have $V(G)=\{v_1, \ldots, v_n, u_1,\ldots,u_n,z\}$, where for each $i\in\{1,\ldots,n\}$,
$u_i$ is adjacent to $z$, $v_i$, and all vertices $v_j$ such that $(i,j)\in E(\vec{K})$.
Let $\vec{H}$ be the partial orientation of $G$ where for $i\in\{1,\ldots,n\}$, the edge $v_iu_i$
for $i\in\{1,\ldots,n\}$ is directed towards $u_i$, and the edge $u_iz$ is directed towards $z$.
Note that for any distinct $i,j\in\{1,\ldots,n\}$, we have $(i,j)\in E(\vec{K})$ or $(j,i)\in E(\vec{K})$,
and thus the path $v_iu_iv_j$ or $v_ju_jv_i$ has the first edge directed towards its middle vertex.
Consequently, $\vec{H}$ is a weak $2$-guidance system for $G$ of maximum outdegree one.
Moreover, any $2$-independent set in $G$ contains at most one of the vertices $\{v_1,\ldots,v_n,z\}$
and at most one of the vertices $\{u_1,\ldots, u_n\}$, and thus $\alpha_2(G)\le 2$.
On the other hand, we have $\gamma_1(G)=\Omega(\log n)$: By replacing each vertex $v_i$ by $u_i$
in an optimal dominating set and possibly adding $z$, we obtain a dominating set $D$ of size at most $\gamma_1(G)+1$
containing none of the vertices $v_1$, \ldots, $v_n$, and to dominate these vertices,
observe that with high probability $D$ needs to contain $\Omega(\log n)$ of the vertices $u_1$, \ldots, $u_n$.
\end{example}

However, we can solve this issue by adding an additional obstruction.  An \emph{$(r,k)$-halfgraph}
in a graph $G$ is a sequence $u_1$, \ldots, $u_k$, $v_1$, \ldots, $v_k$ of vertices of $G$ such that
for every $i,j\in \{1,\ldots,k\}$,
\begin{itemize}
\item if $j<i$, then the distance between $u_i$ and $v_j$ in $G$ is greater than $r$, and
\item if $j\ge i$, then the distance between $u_i$ and $v_j$ in $G$ is exactly $r$.
\end{itemize}
We say that a graph is \emph{$(r,k)$-stable} if it does not contain any $(r,k)$-halfgraph.

\begin{theorem}\label{thm-wguidac}
For all positive integers $r$, $k$, and $c\ge 2$, there exists a constant $b$ and a polynomial-time algorithm
that, given an $(r,k)$-stable graph $G$ together with its weak $2r$-guidance system $\vec{H}$ of maximum outdegree at most $c$,
returns an $r$-dominating set $D$ and a $2r$-independent set $A$ in $G$ such that $|D|\le b|A|$.
\end{theorem}
\begin{proof}
Let $D$ and $A'$ be the sets of vertices of $G$ obtained as follows.  We initialize $D\colonequals \emptyset$ and $A'\colonequals \emptyset$.
As long as $D$ is not an $r$-dominating set, we choose a vertex $x$ at distance greater than $r$ from $D$ arbitrarily,
we add $x$ to $A'$, and we add $x$ and all vertices reachable in $\vec{H}$ from $x$ by directed paths of length at most $r$ to $D$.
At the end, $D$ is an $r$-dominating set and $|D|\le c^{r+1}|A'|$.

Let $\prec$ be the linear ordering on vertices of $A'$ such that $x\prec y$ when $x$ was added to $A'$ before $y$.
The algorithm above enforces the following property ($\dag$): If $x\prec y$, then every vertex reachable from $x$
by a directed path in $\vec{H}$ of length at most $r$ is at distance greater than $r$ from $y$.

Let $\sigma(1)=0$ and for $p=2,\ldots,k$, let $\sigma(p)=c^{2r+1}(\sigma(p-1)+1)$.
The set $A'$ is not necessarily $2r$-independent, however it has the following property:
If $S\subseteq A'$ consists of vertices pairwise at distance at most $2r$ from one another, then $|S|\le \sigma(k+1)$.
To prove this, we will show a stronger claim.  For a positive integer $p\le k+1$, a \emph{$p$-halfgraph extension of $S$}
is a sequence $u_p$, \ldots, $u_k$, $v_p$, \ldots, $v_k$ of vertices of $G$
such that for $i=p,\ldots,k$,
\begin{itemize}
\item[(i)] $u_i\in A'$, $u_i\prec u_{i+1}$ if $i<k$, and $s\prec u_i$ for every $s\in S$.
\item[(ii)] $\vec{H}$ contains a directed path from $u_i$ to $v_i$ of length exactly $r$,
\item[(iii)] the distance between $u_i$ and $v_j$ in $G$ is exactly $r$ for every $j\in\{i,\ldots,k\}$, and
\item[(iv)] the distance between $v_i$ and $s$ is exactly $r$ for every $s\in S$.
\end{itemize}
We will prove by induction on $p$ that if there exists a $p$-halfgraph extension of $S$, then $|S|\le \sigma(p)$;
$|S|\le \sigma(k+1)$ then follows, since an empty sequence trivially forms a $(k+1)$-halfgraph extension of $S$.
For $p=1$, note that if $1\le j<i\le k$, then $u_j\prec u_i$ by (i), and 
(ii) and ($\dag$) imply that the distance between $v_j$ and $u_i$ in $G$ is greater
than $r$. Together with (iii), this implies that $G$ contains an $(r,k)$-halfgraph,
which is a contradiction.  That is, the case $p=1$ can never occur and the conclusion $|S|\le \sigma(1)$ holds trivially.

Suppose now that $p\ge 2$ and that the claim holds for $p-1$.  If $S=\emptyset$, then $|S|\le \sigma(p)$ holds.  Otherwise,
let $u_{p-1}$ be the last vertex of $S$ in the ordering $\prec$.  Since the distance between any vertices of $S$
is at most $2r$ and $\vec{H}$ is a weak $2r$-guidance system, for each $s\in S\setminus\{u_{p-1}\}$, there
exists a shortest path $P_s$ in $G$ between $u_{p-1}$ and $s$ directed in $\vec{H}$ towards one of its edges.
Let $Q_s$ be the longest initial segment of $P_s$ directed away from $u_{p-1}$.  By the choice of $u_{p-1}$,
we have $s\prec u_{p-1}$, and thus ($\dag$) implies that the part of $P_s$ directed away from $s$ has length
at most $r-1$, and consequently $|E(Q_s)|\ge r$.

For any directed path $Q$ in $\vec{H}$ starting in $u_{p-1}$ of length between $r$ and $2r$,
let $S_Q$ be the set of vertices $s\in S\setminus \{u_{p-1}\}$ such that $Q_s=Q$.  The preceding argument shows that $S$ is the union of the
sets $S_Q$ over all such paths, and thus we can fix $Q$ such that $|S_Q|\ge |S|/c^{2r+1}$.  If $|S_Q|\le 1$,
then $|S|\le 2^{2r+1}\le \sigma(p)$, as required.  Hence, suppose that $|S_Q|\ge 2$.
Let $v_{p-1}$ be the final vertex of $Q$ and let $s_Q$ be the first vertex of $S_Q$ in the ordering $\prec$.
Consider any vertex $s'\in S_Q\setminus\{s_Q\}$.  Note that $G$ contains a path of length
at most $2r-|E(Q)|\le r$ from $s_Q$ to $v_{p-1}$ with all but possibly the last edge directed away from $s_Q$ in $\vec{H}$,
and since $s_Q\prec s'$ by the choice of $s_Q$, ($\dag$) implies that $s'$ is at distance at least $r$ from $v_{p-1}$.
Since $s'$ is also at distance at most $2r$ from $u_{p-1}$ through a shortest path whose initial segment is $Q$,
$s'$ is at distance at most $2r-|E(Q)|\le r$ from $v_{p-1}$.  We conclude that $|E(Q)|=r$
and all vertices of $S_Q\setminus \{s_Q\}$ are at distance exactly $r$ from $v_{p-1}$.
Therefore, $u_{p-1}$, \ldots, $u_k$, $v_{p-1}$, \ldots, $v_k$ is a $(p-1)$-halfgraph extension
of $S_Q\setminus\{s_Q\}$, and $|S_Q\setminus\{s_Q\}|\le \sigma(p-1)$ by the induction hypothesis.  But then
$|S|\le c^{2r+1}|S_Q|\le c^{2r+1}(\sigma(p-1)+1)=\sigma(p)$.

Let $F$ be the auxiliary graph with $V(F)=A'$ and with distinct vertices $u,v\in A'$ adjacent
if the distance between them in $G$ is at most $2r$.  We claim that each vertex of $F$ has
at most $c^{2r+1}\sigma(k+1)$ neighbors that precede it in the ordering $\prec$.  Indeed, let
$N$ be the set of such neighbors of a vertex $u\in A'$, and for each directed path $Q$ in $\vec{H}$
starting in $u$ of length between $r$ and $2r$,
let $N_Q$ consist of the vertices $v\in N$ such that $Q$ is the maximal initial directed
segment of a shortest path from $u$ to $v$ in $G$ which is directed towards one of its edges by $\vec{H}$.
As in the preceding part of the proof, note that ($\dag$) and the fact that $\vec{H}$ is
a weak $2r$-guidance system implies that $N$ is the union of the sets $N_Q$ over such paths,
and thus we can fix such a path $Q$ for which $|N_Q|\ge |N|/c^{2r+1}$.
However, the vertices of $N_Q$ are at distance at most $2r-|E(Q)|\le r$ from the
final vertex of $Q$, and thus they are pairwise at distance at most $2r$ from one another.
Consequently, $|N_Q|\le \sigma(k+1)$, and $|N|\le c^{2r+1}\sigma(k+1)$.

We conclude that $F$ is $c^{2r+1}\sigma(k+1)$-degenerate, and thus it is $(c^{2r+1}\sigma(k+1)+1)$-colorable
and has an independent set $A$ of size at least
$$\frac{|A'|}{c^{2r+1}\sigma(k+1)+1}\ge \frac{|D|}{c^{r+1}(c^{2r+1}\sigma(k+1)+1)}.$$
By the construction of $F$, $A$ is a $2r$-independent set in $G$.
Therefore, the theorem holds with $b=c^{r+1}(c^{2r+1}\sigma(k+1)+1)$.
\end{proof}

By the results of Adler and Adler~\cite{adl}, for any structurally nowhere-dense graph class $\GG$
and every $r$, there exists $k$ so that all graphs in $\GG$ are $(r,k)$-stable.
In combination with Corollary~\ref{cor-maina}, we have the following consequence.

\begin{corollary}\label{cor-apxbe}
For any class $\GG$ with structurally bounded expansion and for any positive integer $r$,
there exists a constant $b$ and a polynomial-time randomized algorithm that, given a graph $G\in\GG$
with probability at least $1/2$ returns an $r$-dominating set $D$ and a $2r$-independent set $A$ in $G$ such that $|D|\le b|A|$.
\end{corollary}

\section{Graph classes without bounded outdegree weak guidance systems}\label{sec-lb}

To better understand obstructions to the existence of weak $r$-guidance systems of bounded maximum outdegree,
it is natural to consider the dual of the linear program from the proof of Lemma~\ref{lem-sollp}, which can be reformulated as follows.
For $uz\in E(G)$, let $R_r(u,z)$ be the set of vertices $v\in V(G)$ such that the distance between $u$ and $v$ is between $2$ and $r$
and $z$ lies on a shortest path from $u$ to $v$ in $G$; i.e., $z\in \gate{G}{u}{v}$.

\begin{lemma}\label{lemma-dual}
Let $G$ be a graph and let $r$ be a positive integer.  Let $c$ be the solution to the following
optimization problem:
\begin{align*}
y_{uv}&\ge 0&\text{ for every $u,v\in V(G)$ at distance between $2$ and $r$}\\
x_u&=\max_{z:uz\in E(G)} \sum_{v\in R_r(u,z)} y_{uv}&\text{ for every $u\in V(G)$}\\
\text{maximize }&\frac{\sum_{uv:2\le d_G(u,v)\le r} y_{uv}}{\sum_{v\in V(G)} x_v}
\end{align*}
Then every fractional or weak $r$-guidance system in $G$ has maximum outdegree at least $c$.
\end{lemma}
\begin{proof}
The dual of the linear program from the proof of Lemma~\ref{lem-sollp} is
\begin{align*}
x_u&\ge 0&\text{ for every $u\in V(G)$}\\
y_{uv}&\ge 0&\text{ for every $u,v\in V(G)$ at distance between $2$ and $r$}\\
\sum_{u\in V(G)} x_u&=1\\
\sum_{v\in R_r(u,z)} y_{uv}&\le x_u&\text{ for every $(u,z)$ s.t. $uz\in E(G)$}\\
\text{maximize }&\sum_{uv:2\le d_G(u,v)\le r} y_{uv}
\end{align*}
This is equivalent to the optimization problem from the statement of the lemma.
Hence, its solution $c$ provides a lower bound on the maximum outdegree of a fractional $r$-guidance system in $G$,
and by Observation~\ref{obs-tofra} also a lower bound on the maximum outdegree of a weak $r$-guidance system in $G$.
\end{proof}

As an example, this easily shows that no good weak guidance systems exist for graphs of girth at least five and large maximum average degree
(the \emph{maximum average degree} of a graph is the maximum of the average degrees of its subgraphs).
\begin{lemma}\label{lemma-girth}
Let $G$ be a graph of girth at least five and maximum average degree $d\ge 2$.
Every fractional or weak $2$-guidance system in $G$ has maximum outdegree at least $d/2$.
\end{lemma}
\begin{proof}
Let $Z\subseteq V(G)$ be a smallest set such that $G[Z]$ has average degree~$d$.  Since $d\ge 2$,
every vertex of $G[Z]$ has degree at least two, since deleting vertices of degree at most one would
not decrease the average degree.

Since $G$ has girth at least $5$, any vertices $u,v\in Z$ at distance two in $G[Z]$ have a unique common neighbor $z\in Z$;
we define
$$y_{uv}=\frac{1}{\deg_{G[Z]} z-1}.$$
For any pair $u,v\in V(G)$ of vertices at distance two in $G$ such that $\{u,v\}\not\subseteq Z$ or the common neighbor of $u$ and $v$
does not belong to $Z$, we define $y_{uv}=0$.  For any edge $uz$ of $G$, if $\{u,z\}\subseteq Z$, then
we have $|R_2(u,z)\cap Z|=\deg_{G[Z]} z-1$, and thus
$$\sum_{v\in R_2(u,z)} y_{uv}=1;$$
while if $\{u,z\}\not\subseteq Z$, then 
$$\sum_{v\in R_2(u,z)} y_{uv}=0.$$
Therefore,
$$x_u=\max_{z:uz\in E(G)} \sum_{v\in R_2(u,z)} y_{uv}=1$$
for $u\in Z$ and $x_u=0$ for $u\in V(G)\setminus Z$, and
$$\frac{\sum_{uv:d_G(u,v)=2} y_{uv}}{\sum_{u\in V(G)} x_u}=\frac{\frac{1}{2}\cdot \sum_{u\in Z} \sum_{z:uz\in E(G[Z])} \sum_{v\in R_2(u,z)} y_{uv}}{|Z|}=\frac{|E(G[Z])|}{|Z|}=d/2.$$
The claim now follows from Lemma~\ref{lemma-dual}.
\end{proof}

This shows that weak guidance systems can be qualitatively different from guidance systems only in graphs of girth at most four.
\begin{corollary}
Let $G$ be a graph of girth at least five.  For any $r\ge 2$, if $G$ admits a weak $r$-guidance system of maximum outdegree at most $c$,
then $G$ also admits an $r$-guidance system of maximum outdegree at most $3c$.
\end{corollary}
\begin{proof}
By Lemma~\ref{lemma-girth}, $G$ has maximum average degree at most $2c$, and thus $G$ is $2c$-degenerate.
The claim then follows by Observation~\ref{obs-degen}.
\end{proof}

Next, we consider the class of \emph{split graphs}.  A graph $G$ is a split graph if there exists a partition $(A,B)$ of its
vertex set where $A$ is a clique and $B$ is an independent set.
\begin{lemma}\label{lemma-split}
For every $n$ such that $n$ is a power of a prime, there exists a split graph $G_n$ with
$2(n^2+n+1)$ vertices such that every fractional or weak $2$-guidance system in $G$ has maximum outdegree at least $(n+1)/2$.
\end{lemma}
\begin{proof}
It is well-known that whenever $n$ is a power of prime, there exists a finite projective plane $B$ of order $n$, i.e.,
a system of $n^2+n+1$ subsets of the set $A=[n^2+n+1]$ with the property that
\begin{itemize}
\item[(i)] $|p_1\cap p_2|=1$ for every distinct $p_1,p_2\in B$ and
\item[(ii)] every element of $A$ belongs to exactly $n+1$ sets from $B$.
\end{itemize}
Let $G_n$ be the graph with vertex set $A\cup B$, vertices in $A$ forming a clique, vertices in $B$ forming an independent
set, and vertices $z\in A$ and $p\in B$ adjacent iff $z\in p$.  Note that distinct vertices of $B$ are at distance two
in $G_n$ by (i), and that for each $p\in B$ and $z\in p$, $|R_2(p,z)\cap B|=n$ by (ii).  Therefore, defining $y_{p_1p_2}=1$
for any distinct $p_1,p_2\in B$ and $y_{uv}=0$ for any other pair $u,v$ of vertices of $G_n$,
we have
$$x_p=\max_{z:z\in p} \sum_{p'\in R_2(p,z)} y_{pp'}=n$$
for $p\in B$ and $x_z=0$ for $z\in A$.
Therefore,
$$\frac{\sum_{uv:d_{G_n}(u,v)=2} y_{uv}}{\sum_{u\in V(G_n)} x_u}=\frac{\binom{|B|}{2}}{|B|n}=\frac{|B|-1}{2n}=\frac{n+1}{2}.$$
The claim now follows from Lemma~\ref{lemma-dual}.
\end{proof}
Let us remark that split graphs are a special case of \emph{chordal graphs} (graphs with no induced cycle of length at least four),
and thus chordal graphs do not in general admit weak guidance systems of bounded maximum outdegree.

Finally, let us consider the graphs of bounded \emph{clique-width}.  A \emph{$k$-labeled graph} is a graph
where each vertex is assigned a label from $[k]$ (several vertices can have the same label,
and not all labels must be used).  A $k$-labeled graph $G$ is \emph{constructible} if it is obtained
by a finite number of applications of the following rules:
\begin{itemize}
\item $|V(G)|=1$, or
\item $G$ is the disjoint union of at least two constructible $k$-labeled graphs, or
\item $G$ is obtained from a constructible $k$-labeled graph $G'$ by, for some $i,j\in[k]$,
changing all labels $i$ to $j$, or
\item $G$ is obtained from a constructible $k$-labeled graph $G'$ by, for some $i,j\in[k]$,
adding all edges between vertices with labels $i$ and $j$.
\end{itemize}
We say a graph has \emph{clique-width} at most $k$ if we can assign labels to its vertices so that the resulting $k$-labeled graph is constructible.
Graphs with bounded shrub-depth also have bounded clique-width (or equivalently, bounded rank-width); indeed,
they can be viewed as graphs of bounded clique-width where the corresponding operation tree has bounded depth.
It is natural to ask whether Lemma~\ref{lemma-orshrub} extends to graphs of bounded clique-width.
We show that this is not the case, even for weak $2$-guidance systems.
\begin{lemma}\label{lemma-cwbad}
For every $d\ge 0$ and $a\ge \max(2,2d-1)$, there exists a constructible 6-labeled graph $H_{d,a}$
with half its vertices labeled $1$ and half its vertices labeled $2$,
such that
\begin{itemize}
\item[(i)] $|V(H_{d,a})|\le 8a^d-6$ and
\item[(ii)] for every partial orientation $\vec{G}$ of $H_{d,a}$ of maximum outdegree less than $d$, there
exist vertices $u$ and $v$ of labels $1$ and $2$, respectively, at distance exactly two, such that
for every common neighbor $x$ of $u$ and $v$, we have $(u,x),(v,x)\not\in E(\vec{G})$.
\end{itemize}
\end{lemma}
\begin{proof}
For $d=0$, we can let $H_{0,a}=K_2$ with one vertex labeled $1$ and the other vertex labeled $2$.  Suppose we already constructed $H_{d-1,a}$,
and let us show how to inductively obtain $H_{d,a}$.  First, let $H'_{d-1,a}$ be the graph obtained from $H_{d-1,a}$ by adding vertices $v_3$ and $v_4$
with labels $3$ and $4$ and adding all edges between vertices with labels $1$ and $4$ and between vertices with labels $2$ and $3$.
Next, we form the disjoint union of $a$ copies of $H'_{d-1,a}$.
Then we add two vertices $v_5$ and $v_6$ with labels $5$ and $6$, and all edges between vertices with labels $i$ and $i+2$ for $i\in\{3,4\}$.
Finally, we relabel vertices with labels $3$ and $5$ to label $1$ and vertices with labels $4$ and $6$ to label $2$.

The construction uses only $6$ labels, and thus $H_{d,a}$ is a constructible 6-labeled graph.  Moreover,
$$|V(H_{d,a})|=a(|V(H_{d-1,a})+2)+2\le a(8a^{d-1}-4)+2\le 8a^d-6,$$
where the last inequality holds since $a\ge 2$.
Consider any partial orientation $\vec{G}$ of $H_{d,a}$ of maximum outdegree less than $d$.
Since $v_5$ and $v_6$ have outdegree less than $d$, for one of the $a\ge 2d-1$ copies of $H'_{d-1,a}$ in $H_{d,a}$, denoted by $F'$,
we have $(v_i,v)\not\in \vec{G}$ for every $i\in\{5,6\}$ and $v\in V(F')$.  Let $F$ be the copy of $H_{d-1,a}$
in $F'$.  Suppose that for any two vertices $u$ and $v$ of $F$ of labels $1$ and $2$,
respectively, at distance exactly two in $H_{d,a}$, there exists a common neighbor $x$ of $u$ and $v$ in $H_{d,a}$ such that
$(u,x)\in E(\vec{G})$ or $(v,x)\in E(\vec{G})$.  The construction of $H'_{d-1,a}$ and $H_{d,a}$ ensures that such a common neighbor $x$ necessarily belongs to $F$,
as we did not add any vertex adjacent both to vertices with label $1$ and with label $2$.  Hence, by the induction hypothesis, the restriction
of $\vec{G}$ to $F$ has maximum outdegree at least $d-1$.  Let $u$ be a vertex of $F$ with at least $d-1$ outneighbors in $\vec{G}$ belonging to $F$.
By symmetry, we can assume $u$ has label $1$.  Since $\vec{G}$ has maximum outdegree less than $d$, we have $(u,v_4)\not\in E(\vec{G})$.
Moreover, by the choice of $F'$, we have $(v_6,v_4)\not\in E(\vec{G})$.  Note that $v_6$ has label $2$ in $H_{d,a}$ and the copy of $v_4$ in $F$
is the only common neighbor of $u$ and $v_6$ in $H_{d,a}$. This shows that $H_{d,a}$ satisfies the property (ii).
\end{proof}

By Lemma~\ref{lemma-cwbad}, letting $n=|V(H_{d,2d-1})|$, we conclude that any weak $2$-guidance system in $H_{d,2d-1}$, a graph of clique-width at most $6$,
has maximum outdegree at least $d=\Omega(\log n/\log\log n)$.  As we will see in Lemma~\ref{lemma-cutcom}, this is nearly tight.  
Before that, let us remark that a similar bound also applies to fractional $2$-guidance systems, which follows from Lemma~\ref{lemma-dual}:
For the purpose of the analysis, let us define both vertices of $H_{0,a}$ to be \emph{foundational}, and when constructing
$H_{d,a}$, we let the foundational vertices be exactly the foundational vertices in the copies of $H_{d-1,a}$;
then, we consider the $y$-weights defined inductively for each copy of $H_{d-1,a}$, and additionally set $y_{v_iz}=1$ for each $i\in\{5,6\}$ and each foundational vertex $z$ at distance two from $v_i$.
Letting $n_d=2a^d$ be the number of foundational vertices of $H_{d,a}$, this results in $x_{v_i}=|n_{d-1}|/2$; and moreover, $x_z=1$ for every foundational vertex $z$.  
Hence, the lower bound we obtain by Lemma~\ref{lemma-dual} is at least
$$\frac{an_{d-1}+a^2n_{d-2}+\ldots+a^dn_0}{(n_{d-1}+an_{d-2}+\ldots+a^dn_0)+n_d}=\frac{2da^d}{2da^{d-1}+2a^d}=\frac{da}{d+a}=\frac{2}{3}d$$
for $a=2d$.

On the positive side, we show that graphs of bounded clique-width admit weak guidance systems of logarithmic outdegree.
Let us start by a useful observation.  Suppose $(A,B)$ is a partition
of the vertex set of a graph $G$.  For $u,v\in V(G)$, we write $u\equiv_{(A,B)} v$ if either $u,v\in A$ and $u$ and $v$ have the same neighbors in $B$,
or $u,v\in B$ and $u$ and $v$ have the same neighbors in~$A$.

\begin{lemma}\label{lemma-cutcom}
Let $r$ be a positive integer or $\infty$.
Suppose $(A,B)$ is a partition of the vertex set of a graph $G$ and $\equiv_{(A,B)}$ has $k$ equivalence classes.
If $G[A]$ and $G[B]$ have a weak $r$-guidance system of maximum outdegree at most $c$, then $G$ has a weak $r$-guidance
system of maximum outdegree at most $c+k$.
\end{lemma}
\begin{proof}
Let $\vec{H}_A$ and $\vec{H}_B$ be weak $r$-guidance systems of maximum outdegree at most $c$ in $G[A]$ and $G[B]$, respectively.
Let $\vec{H}$ consist of $\vec{H}_A\cup \vec{H}_B$ and the following edges:
For each $u\in V(G)$ and each equivalence class $C$ of $\equiv_{(A,B)}$ intersecting the component of $G$ containing $u$,
choose a vertex $u'_C$ in $C$ nearest to $u$ in $G$ and a vertex $u_C\in \gate{G}{u}{u'_C}$ arbitrarily, and add the edge
$(u,u_C)$.  Clearly, $\vec{H}$ has maximum outdegree at most $c+k$.

Consider now any vertices $u,v\in V(G)$ at distance $\ell$, where $2\le \ell\le r$, and let $P$ be a shortest path
between $u$ and $v$ in $G$.  If an edge of $P$ incident with $u$ or $v$ belongs to $G[A]\cup G[B]$, switch the names
of vertices $u$ and $v$ if necessary so that such an edge is incident with $u$.  By symmetry, we can assume $u\in A$.
If $P\subseteq G[A]$, then by Observation~\ref{obs-char}, $\vec{H}_A$ (and thus also $\vec{H}$) contains
an edge directed from $u$ to $\gate{G[A]}{u}{v}\subseteq \gate{G}{u}{v}$ or an edge directed from $v$ to $\gate{G[A]}{v}{u}\subseteq \gate{G}{v}{u}$.
Hence, suppose that $P\not\subseteq G[A]$.

If the first edge of $P$ is contained in $G[A]$, then let $P'$ be the longest initial segment
of $P$ contained in $G[A]$.  If the first edge of $P$ is not contained in $G[A]$, then let $P'$ be the longest initial segment
of $P$ contained in $G[B\cup\{u\}]$.  Let $C$ be the equivalence class of $\equiv_{(A,B)}$ containing the last vertex $z$ of $P'$.
Note that $z\neq v$: In the first case, this is because $P$ is not contained in $G[A]$.  In the second case, this
is because $|E(P)|=\ell\ge 2$ and the choice of the names of the vertices $u$ and $v$ implies that the last edge of $P$ is not contained in $G[B]$.
Since $u'_C$ is a nearest vertex from $u$ in $C$, $u'_C$ is at distance at most $|E(P')|$ from $u$ in $G$.
Moreover, $u'_C$ is in the same equivalence class of $\equiv_{(A,B)}$ as $z$, and thus $u'_C$ is adjacent to the vertex following $z$ in $P$.
Hence, $u_C\in\gate{G}{u}{v}$ and $\vec{H}$ contains the edge $(u,u_C)$.

Observation~\ref{obs-char} then implies that $\vec{H}$ is a weak $r$-guidance system in $G$.
\end{proof}

We combine this with the following well-known fact about clique-width.
\begin{observation}
If $G$ is a graph with $n$ vertices and clique-width at most $k$, then there exists a partition $(A,B)$ of vertices of $G$
such that $|A|,|B|\le\tfrac{2}{3}n$ and $\equiv_{(A,B)}$ has at most $2k$ equivalence classes.
\end{observation}

Since any induced subgraph of a graph of clique-width at most $k$ also has clique-width at most $k$, we obtain the following
consequence.

\begin{corollary}\label{lemma-cwgood}
For every $k\ge 0$, every $n$-vertex graph of clique-width at most $k$ has a partial orientation $\vec{H}$
of maximum outdegree $O(k\log n)$ such that $\vec{H}$ is a weak $\infty$-guidance system.
\end{corollary}

\section{Conclusions}

As we have shown, some interesting graph classes admit weak guidance systems of bounded maximum
outdegree, including
\begin{itemize}
\item interval graphs,
\item classes with structurally bounded expansion, and
\item distance powers of graphs with bounded outdegree weak guidance systems.
\end{itemize}
However, we do not have an exact characterization of the graph classes with this property.

\begin{problem}
Characterize hereditary graph classes $\GG$ such that for every positive integer $r$, every graph from $\GG$ admits a weak $r$-guidance system
of bounded maximum outdegree.
\end{problem}

We have also exhibited several graph classes that only admit weak guidance systems whose outdegree
grows slowly with the number of vertices of the graph, in particular
\begin{itemize}
\item structurally nowhere-dense classes, and
\item graphs of bounded clique-width.
\end{itemize}
Again, we do not have a good description of the graph classes with this property.

\begin{problem}
Characterize hereditary graph classes $\GG$ such that for every positive integer $r$, every graph $G\in \GG$ admits a weak $r$-guidance system
of maximum outdegree at most $|V(G)|^{o(1)}$.
\end{problem}

In sparse graphs, guidance systems and related notions (such as the generalized coloring numbers) have various algorithmic
and structural applications.  We suspect that similar applications can be found for weak guidance systems as well,
generalizing them to dense graphs; we demonstrated this on the example of approximation algorithms
for distance domination and independence number.

\section*{Acknowledgments}

I would like to thank Abhiruk Lahiri and Ben Moore for useful discussions of the subject.

\bibliographystyle{siam}
\bibliography{../data.bib}

\begin{thebibliography}{10}

\bibitem{adl}
{\sc H.~Adler and I.~Adler}, {\em Interpreting nowhere dense graph classes as a
  classical notion of model theory}, European Journal of Combinatorics, 36
  (2014), pp.~322--330.

\bibitem{dreier2021lacon}
{\sc J.~Dreier}, {\em Lacon-and shrub-decompositions: a new characterization of
  first-order transductions of bounded expansion classes}, in 36th Annual
  ACM/IEEE Symposium on Logic in Computer Science (LICS), IEEE, 2021,
  pp.~1--13.

\bibitem{dreier2022treelike}
{\sc J.~Dreier, J.~Gajarsk{\'y}, S.~Kiefer, M.~Pilipczuk, and
  S.~Toru{\'n}czyk}, {\em Treelike decompositions for transductions of sparse
  graphs}, arXiv, 2201.11082 (2022).

\bibitem{apxdomin}
{\sc Z.~Dvo{\v{r}}{\'a}k}, {\em Constant-factor approximation of domination
  number in sparse graphs}, European Journal of Combinatorics, 34 (2013),
  pp.~833--840.

\bibitem{dvovrak2018induced}
\leavevmode\vrule height 2pt depth -1.6pt width 23pt, {\em Induced subdivisions
  and bounded expansion}, European Journal of Combinatorics, 69 (2018),
  pp.~143--148.

\bibitem{dvorlah}
{\sc Z.~Dvo{\v{r}}{\'a}k and A.~Lahiri}, {\em Approximation schemes for bounded
  distance problems on fractionally treewidth-fragile graphs}, arXiv,
  2105.01780 (2021).

\bibitem{gajarsky2020new}
{\sc J.~Gajarsk{\'y}, P.~Hlin{\v{e}}n{\`y}, J.~Obdr{\v{z}}{\'a}lek,
  D.~Lokshtanov, and M.~Ramanujan}, {\em A new perspective on {FO} model
  checking of dense graph classes}, ACM Transactions on Computational Logic
  (TOCL), 21 (2020), pp.~1--23.

\bibitem{gajarsky2020first}
{\sc J.~Gajarsk{\'y}, S.~Kreutzer, J.~Ne{\v{s}}et{\v{r}}il, P.~O.~D. Mendez,
  M.~Pilipczuk, S.~Siebertz, and S.~Toru{\'n}czyk}, {\em First-order
  interpretations of bounded expansion classes}, ACM Transactions on
  Computational Logic (TOCL), 21 (2020), pp.~1--41.

\bibitem{shrub}
{\sc R.~Ganian, P.~Hlin\v{e}n{\'y}, J.~Ne\v{s}et\v{r}il, J.~Obdr\v{z}{\'a}lek,
  and P.~O. de~Mendez}, {\em {Shrub-depth: Capturing Height of Dense Graphs}},
  Logical Methods in Computer Science, 15 (2019).

\bibitem{KowKur}
{\sc L.~Kowalik and M.~Kurowski}, {\em Oracles for bounded length shortest
  paths in planar graphs}, ACM Trans. Algorithms, 2 (2006), pp.~335--363.

\bibitem{nevsetvril2021rankwidth}
{\sc J.~Ne{\v{s}}et{\v{r}}il, P.~O.~d. Mendez, M.~Pilipczuk, R.~Rabinovich, and
  S.~Siebertz}, {\em Rankwidth meets stability}, in Proceedings of the 2021
  ACM-SIAM Symposium on Discrete Algorithms (SODA), SIAM, 2021, pp.~2014--2033.

\bibitem{nesbook}
{\sc J.~Ne{\v{s}}et\v{r}il and P.~{Ossona de Mendez}}, {\em Sparsity (Graphs,
  Structures, and Algorithms)}, vol.~28 of Algorithms and Combinatorics,
  Springer, 2012.

\bibitem{pach2011combinatorial}
{\sc J.~Pach and P.~K. Agarwal}, {\em Combinatorial geometry}, John Wiley \&
  Sons, 2011.

\bibitem{numtypes}
{\sc M.~Pilipczuk, S.~Siebertz, and S.~Toru\'{n}czyk}, {\em On the number of
  types in sparse graphs}, in Proceedings of the 33rd Annual ACM/IEEE Symposium
  on Logic in Computer Science, LICS'18, ACM, 2018, pp.~799--808.

\end{thebibliography}

\end{document}